\documentclass[11pt,a4paper,reqno]{amsart}
\usepackage{graphicx}
\usepackage{epsfig}
\usepackage{cite}
\usepackage[english]{babel}
\usepackage{amsmath}
\usepackage{amstext}
\usepackage{amssymb}
\usepackage{mathtools}
\usepackage{lipsum}
\usepackage{amsfonts}
\usepackage{graphicx}
\usepackage{epstopdf}
\usepackage{algorithmic}
\ifpdf
  \DeclareGraphicsExtensions{.eps,.pdf,.png,.jpg}
\else
  \DeclareGraphicsExtensions{.eps}
\fi

 \newtheorem{lemma}{Lemma}
 
 \newtheorem{definition}{Definition}

 \newtheorem{remark}{Remark}
\newtheorem{assumption}{Assumption}
\newtheorem{theorem}{Theorem}

\def\rL{\mathcal{L}}

\def\R{\mathbb{R}}

\def\rP{\mathbb{P}}

\def\supp{\mathop{\rm supp}}

\def\Car{\mathop{\rm Car}}

\def\L{\mathop{\rm L}}

\def\Proj{\mathop{\rm Proj}}

\def\A{{\mathcal A}}
\def\B{{\mathcal B}}

\def\G{{\mathcal G}}
\def\C{{\mathcal C}}

\def\P{{\mathcal P}}

\def\X{{\mathcal X}}

\def\Y{{\mathcal Y}}
\def\M{{\mathcal M}}

\def\cE{\mathbb{E}}

\def\U{{\mathcal U}}

\def\by{{\bf y}}
\def\bu{{\bf u}}

\def\ess{{\mathrm{ess}}}

\def\sX{{\mathsf X}}
\def\sY{{\mathsf Y}}

\def\sY{{\mathsf Y}}
\def\sM{{\mathsf M}}
\def\sE{{\mathsf E}}

\def\sU{{\mathsf U}}

\allowdisplaybreaks

\begin{document}
\setcounter{page}{1}
\bigskip
\bigskip
\title[N. Saldi: Common Information Approach for Static Teams with Polish Spaces] {Common Information Approach for Static Team Problems with Polish Spaces and Existence of Optimal Policies}
\author{N. Saldi$^1$
 }
\thanks{$^1$Department of Mathematics, Bilkent University, Cankaya, Ankara, Turkey,
\\ \indent\,\,\,e-mail: naci.saldi@bilkent.edu.tr
\\ \indent
  \em \,\,\,}

\begin{abstract}
In this paper, we demonstrate the existence of team-optimal strategies for static teams under observation-sharing information structures. Assuming that agents can access shared observations, we begin by converting the team problem into an equivalent centralized stochastic control problem through the introduction of a topology on policies. We subsequently apply conventional methods from stochastic control to prove the existence of team-optimal strategies. This study expands upon the widely recognized common information approach for team problems, originally designed for discrete scenarios, and adapts it to a more abstract continuous framework. The primary difficulty in this context is to establish the appropriate topology on policies.


\end{abstract}
\maketitle

\smallskip

\section{Introduction}\label{sec1}

Team decision theory has been introduced by Marschak \cite{mar55} to investigate the collective actions of a group of agents operating in a decentralized manner with the aim of optimizing a shared cost function. Radner, as highlighted in \cite{rad62}, established crucial insights into static teams, including the connections between person-by-person optimality and team optimality. Dynamic teams and the characterization and categorization of information structures saw significant progress thanks to Witsenhausen's influential papers, as documented in \cite{wit71,wit75,wit88,WitsenStandard,WitsenhausenSIAM71,Wit68}. In particular, Witsenhausen's well-known counterexample in \cite{Wit68} shed light on the challenges posed by decentralized information structures in such models. For a more extensive overview of team decision theory and an in-depth exploration of the literature, we direct readers to \cite{YukselBasarBook,SaYu22}.

The fundamental distinction between team decision problems and classical centralized decision problems lies in the decentralized information structure. In team decision problems, agents are unable to share their information with each other. This decentralized information setup restricts the application of conventional tools used in centralized decision theory, such as dynamic programming, convex analytical techniques, and linear programming. Consequently, establishing the existence and structure of optimal policies becomes a notably challenging task within the realm of team decision theory.

In the existing literature, team decision problems are typically addressed through three primary approaches, as outlined in \cite{CDCTutorial}: (i) the common information approach \cite{NayyarBookChapter, NayyarMahajanTeneketzis}, (ii) the designer's approach \cite{WitsenStandard, MahajanThesis, MahTen09}, and (iii) the person-by-person approach \cite{rad62, marrad72}. Among these methods, the most effective one is the common information approach.

In the common information approach, it is assumed that agents possess shared information among themselves, which could be in the form of delayed or periodic observation sharing. Consequently, the information held by each agent can be divided into two categories: common information and private information. In other words, there exists a coordinator who observes the common information and shares this information to the other agents. Within the common information approach, the objective is to frame the problem as a centralized stochastic control problem from the coordinator's perspective. By adopting this viewpoint, classical stochastic control techniques like dynamic programming can then be employed to compute the optimal team decision strategy.

The common information approach was initially developed for discrete setup; that is, the state, observation, and action spaces are all finite. The objective of this paper is to broaden the application of this approach to continuous spaces. The primary difficulty in this context is to establish the appropriate topology on policies. 
To achieve this, we introduce a topology on the collection of policies inspired by the topology introduced in \cite[Section 2.4]{BoArGh12}. Indeed, this topology is first introduced and used in \cite{Sal20} to establish the existence of team optimal policies for fully decentralized team problems. In this work, assuming that agents have access to shared observations, our first step involves transforming the team problem into an equivalent centralized stochastic control problem using this topology. Subsequently, we employ conventional techniques from stochastic control to establish the existence of optimal strategies for the team.

The equivalence between the team problem and its centralized counterpart also facilitates the computation of this optimal policy. Nevertheless, due to the continuous nature of the spaces involved, this results in an optimization problem over an infinite-dimensional space. This complexity can be addressed through approximation methods, such as quantization \cite{SaYuLi17}.

In the literature, there is a limited number of findings regarding the existence of team-optimal solutions for classical (non-observation sharing) team problems. To date, papers \cite{GuYuBaLa15,YuSa17} have primarily focused on establishing the existence of optimal policies for static teams and a specific class of sequential dynamic teams. These studies adopt a strategic measure approach, where strategic measures represent the probability measures resulting from policies applied to the product of state space, observation spaces, and action spaces. In this approach, the process begins with the identification of a topology for the set of strategic measures, followed by demonstrating the relative compactness of this set and the lower semi-continuity of the cost function. If the set of strategic measures is proven to be closed, one can then invoke the Weierstrass Extreme Value Theorem to establish the existence of optimal policies. We can also employ strategic measure approach to address team decision problems under observation sharing information structure and establish the existence of optimal policies. In the sequel, we will outline a concise overview of this approach. Consequently, the strategic measure approach can be seen as a complementary method to ours to tackle such information structures.

\subsection{Notation and Conventions}

For a metric space $\sE$, the Borel $\sigma$-algebra is denoted by $\mathcal{E}$.  We let $C_0(\sE)$ and $C_c(\sE)$ denote the set of all continuous real functions on $\sE$ vanishing at infinity and the set of all continuous real functions on $\sE$ with compact support, respectively. For any $g \in C_c(\sE)$, let $\supp(g)$ denote its support. Let $\M(\sE)$ and $\P(\sE)$ denote the set of all finite signed measures and probability measures on $\sE$, respectively. A sequence $\{\mu_n\}$ of finite signed measures on $\sE$ is said to converge with respect to total variation distance (see \cite{HeLa03}) to a finite signed measure $\mu$ if $ \lim_{n\rightarrow\infty} 2\sup_{D \in \mathcal{E}} |\mu_n(D) - \mu(D)|=0$. A sequence $\{\mu_n\}$ of finite signed measures on $\sE$ is said to converge weakly (see \cite{HeLa03}) to a finite signed measure $\mu$ if $\int_{\sE} g d\mu_n \rightarrow \int_{\sE} g d\mu$ for all bounded and continuous real function $g$ on $\sE$. Let $\sE_1$ and $\sE_2$ be two metric spaces. For any $\mu \in \M(\sE_1 \times \sE_2)$, we denote by $\Proj_{\sE_1}(\mu)(\,\cdot\,) \triangleq \mu(\,\cdot\,\times \sE_2)$ the marginal of $\mu$ on $\sE_1$. Let $\sE = \prod_{i=1}^N \sE_i$ be a finite product space. For each $j,k = 1,\ldots,N$ with $k < j$, we denote $\sE^{^{[k:j]}} = \prod_{i=k}^j \sE_i$ and $\sE_{-j} = \prod_{i \neq j} \sE_i$. A similar convention also applies to elements of these sets which will be denoted by bold lower case letters. For any set $D$, let $D^c$ denote its complement. For any random element $w$, $\L(w)$ denotes its distribution. Unless otherwise specified, the term `measurable' will refer to Borel measurability in the rest of the paper.

\section{Intrinsic Model for Sequential teams}\label{sec1}

Team problems with static information structure and with common information has the following components:
\begin{align}
\bigl\{ (\sX, {\mathcal X}), (\sX_0, {\mathcal X}_0),  (\sU_i,{\mathcal U}_i), (\sY_i,{\mathcal Y}_i),\rP, i=1,\ldots,N,\bigr\} \nonumber
\end{align}
where locally compact Polish spaces (i.e., complete and separable metric spaces) $\sX$, $\sX_0$, $\sU_i$, and $\sY_i$ ($i=1,\ldots,N$) endowed with Borel $\sigma$-algebras denote the state space, common information space, and action and observation spaces of Agent~$i$, respectively. Here $N$ is the number of agents and $x \in \sX$ is the state variable. For each $i$, the observations and actions of Agent~$i$ are denoted by $u_i$ and $y_i$, respectively.
The $\sY_i$-valued observation variable for Agent~$i$ is given by $y_i \sim W_i(\,\cdot\,|x)$, where $W_i$ is a stochastic kernel from $\sX$ to $\sY_i$ (see Definition~\ref{stoc-ker}). In addition to above variables, we also assume that there is a random variable $x_0$, which lives in $\sX_0$ and correlated to the state $x$, such that it is a common information to all agents in the team problem. A probability measure $\rP$ denotes the joint law of the random variables $(x,x_0)$. Let $\mu_0$ and $\mu$ denote the marginals of $\rP$ on $\sX_0$ and $\sX$, respectively.

\begin{assumption}\label{as1}
For all $i$, $W_i: \sX \rightarrow \P(\sY_i)$ is continuous with respect to the total variation norm and  $W_i(dy_i|x) = q_i(y_i,x) \, \mu_i(dy_i)$ for some probability measure $\mu_i$ on $\sY_i$.
\end{assumption}

For each $i=1,\ldots,N$, a control strategy $\gamma_i$ for Agent~$i$ is a $\mu_i\otimes\mu_0$-stochastic kernel from $\sY_i \times \sX_0$ to $\sU_i$ (see Definition~\ref{stoc-ker}). Let $\Gamma_i$ denote the set of all equivalence classes of control strategies for Agent~$i$; that is, two control strategies $\gamma_i$ and $\tilde{\gamma}_i$ are equivalent if $\gamma_i=\tilde{\gamma}_i$ $\mu_i\otimes\mu_0$-a.e.. Let ${\bf \Gamma} = \prod_{k} \Gamma_k$.

For any $\underline{\gamma} = (\gamma_1, \cdots, \gamma_N)$, we let the
(expected) cost of the team problem be defined by
\[J(\underline{\gamma}) \triangleq E[c(x,x_0, {\bf y},{\bf u})],\]
for some lower semi-continuous cost function $c: \sX \times \sX_0 \times \prod_i \sY_i \times \prod_i \sU_i \to [0,\infty)$, where ${\bf u} \triangleq (u_1,\ldots,u_N) \sim \prod_{i=1}^N \gamma_i(y_i,x_0)(\,\cdot\,)$ and ${\bf y} \triangleq (y_1,\ldots,y_N)$.

\begin{definition}\label{Def:TB1}
For a given stochastic team problem, a policy (strategy)
${\underline \gamma}^*:=({\gamma_1}^*,\ldots, {\gamma_N}^*)\in {\bf \Gamma}$ is
an {\it optimal team decision rule} if
\begin{equation}
J({\underline \gamma}^*)=\inf_{{{\underline \gamma}}\in {{\bf \Gamma}}}
J({{\underline \gamma}})=:J^*. \nonumber
\end{equation}
The cost level $J^*$ achieved by this strategy is the {\it optimal team cost}.
\end{definition}

\section{Duals of Vector Valued Functions}\label{sec0}

To make the paper as self-contained as possible, in this section we review duality results for vector valued functions that will be used to construct the topology for team policies. 
 
Let $(\sX,\X,\mu)$ be a probability space, where $\sX$ is a Polish space. Let $(\sY,\|\cdot\|_{\sY})$ be a separable Banach space and let $\sY^*$ denote the topological dual of $\sY$ with the induced norm $\|\cdot\|_{\sY^*}$ which turns $\sY^*$ into a Banach space. The duality pairing between any $y \in \sY$ and any $y^* \in \sY^*$ is denoted by $\langle y^*,y \rangle$. Hence, for any $y^* \in \sY^*$, the mapping $\sY \ni y \mapsto \langle y^*,y \rangle \in \R$ is linear and continuous.

We now define Bochner integrable functions from $(\sX,\X,\mu)$ to $\sY$. Similar to the definition of measurable functions, we start with definition of simple functions. A function $f: \sX \rightarrow \sY$ is said to be simple if there exists $y_1,\ldots,y_n \in \sY$ and $E_1,\ldots,E_n \in \X$ such that
$
f(x) = \sum_{i=1}^n y_i 1_{E_i}(x). \nonumber
$
Define the Bochner integral of $f$ with respect to $\mu$ as
$
\int_{\sX} f(x) \, \mu(dx) \triangleq \sum_{i=1}^n y_i \mu(E_i). \nonumber
$
A function $f:\sX \rightarrow \sY$ is  said to be strongly measurable, if there exists a sequence $\{f_n\}$ of simple functions with $\lim_{n\rightarrow\infty} \|f_n(x)-f(x)\|_{\sY} = 0$ $\mu$-a.e.. The strongly measurable function $f$ is Bochner-integrable \cite{DiUh77} if $\int_{\sX} \|f(x)\|_{\sY} \, \mu(dx) < \infty$. In this case, the integral is given by
$
\int_{\sX} f(x) \, \mu(dx) = \lim_{n\rightarrow\infty} \int_{\sX} f_n(x) \, \mu(dx), \nonumber
$
where $\{f_n\}$ is a sequence of simple functions which approximates $f$. Let $L_1\bigl(\mu,\sY\bigr)$ denote the set of all equivalence classes of Bochner-integrable functions from $(\sX,\X,\mu)$ to $\sY$ endowed with the norm
\begin{align}
\|f\|_1 \triangleq \int_{\sX} \|f(x)\|_{\sY} \, \mu(dx). \nonumber
\end{align}
With this norm, $L_1\bigl(\mu,\sY\bigr)$ is a separable Banach space. 

We now identify the topological dual of $L_1\bigl(\mu,\sY\bigr)$ which is denoted by $L_1\bigl(\mu,\sY\bigr)^*$. We start with the definition of $w^*$-measurable functions from $\sX$ to $\sY^*$. A function $\gamma: \sX \rightarrow \sY^*$ is called $w^*$-measurable \cite[p. 18]{CeMe97} if the mapping $\sX \ni x \mapsto \langle \gamma(x), y \rangle \in \R$ is $\X/\B(\R)$-measurable for all $y \in \sY$. Let $\rL\bigl(\mu,\sY^*\bigr)$ denote the set of all equivalence classes of $w^*$-measurable functions. Then, we define the following subset
\begin{align}
&\rL_{\infty}\bigl(\mu,\sY^*\bigr) \triangleq \biggl\{ \gamma \in \rL\bigl(\mu,\sY^*\bigr): \|\gamma\|_{\infty} \triangleq\ess \sup_{x \in \sX} \|\gamma(x)\|_{\sY^*} < \infty \biggr\}, \label{eq1}
\end{align}
where $\ess \sup$ is taken with respect to the measure $\mu$. Then, we have the following theorem.

\begin{theorem}\cite[Theorem 1.5.5]{CeMe97}\label{mainduality}
For any $\gamma \in \rL_{\infty}\bigl(\mu,\sY^*\bigr)$ and $f \in L_1\bigl(\mu,\sY\bigr)$, let
\begin{align}
T_{\gamma}(f) \triangleq \int_{\sX} \langle \gamma(x), f(x) \rangle \, \mu(dx). \nonumber
\end{align}
Then the map $\rL_{\infty}\bigl(\mu,\sY^*\bigr) \ni \gamma \mapsto T_{\gamma} \in L_1\bigl(\mu,\sY\bigr)^*$ is an isometric isomorphism from $\rL_{\infty}\bigl(\mu,\sY^*\bigr)$ to $L_1\bigl(\mu,\sY\bigr)^*$. Hence, we can identify $L_1\bigl(\mu,\sY\bigr)^*$ with $\rL_{\infty}\bigl(\mu,\sY^*\bigr)$. For any $f \in L_1\bigl(\mu,\sY\bigr)$ and $\gamma \in \rL_{\infty}\bigl(\mu,\sY^*\bigr)$, the duality pairing is given by
\begin{align}
\langle\langle \gamma , f \rangle\rangle \triangleq \int_{\sX} \langle \gamma(x), f(x) \rangle \, \mu(dx).\nonumber
\end{align}
\end{theorem}

\begin{remark}
If $\sY^*$ is reflexive or separable, then  $L_1\bigl(\mu,\sY\bigr)^*$ can be identified with set of strongly measurable functions that are $\mu$-essentially bounded \cite[Theorem 4.2.26]{PaWi18}. However, in our case, $\sY^*$ is to be the set of finite signed measures over some locally compact Polish space with total variation norm, which is obviously neither reflexive nor separable. 
\end{remark}

By Theorem~\ref{mainduality}, we equip $\rL_{\infty}\bigl(\mu,\sY^*\bigr)$ with $w^*$-topology induced by $L_1\bigl(\mu,\sY\bigr)$; that is, it is the smallest topology on $\rL_{\infty}\bigl(\mu,\sY^*\bigr)$ for which the mapping
\begin{align}
\rL_{\infty}\bigl(\mu,\sY^*\bigr) \ni \gamma \mapsto \int_{\sX} \langle \gamma(x), f(x) \rangle \, \mu(dx) \in \R \nonumber
\end{align}
is continuous for all $f \in L_1\bigl(\mu,\sY\bigr)$. We write $\gamma_{\lambda} \rightharpoonup^* \gamma$, if $\gamma_{\lambda}$ converges to $\gamma$ in $\rL_{\infty}\bigl(\mu,\sY^*\bigr)$ with respect to $w^*$-topology.

Suppose $G$ is a subset of $\sY^*$ and define
\begin{align}
\rL_{\infty}\bigl(\mu,G\bigr) \triangleq \biggl\{ \gamma \in \rL_{\infty}\bigl(\mu,\sY^*\bigr): \gamma(y) \in G \text{ } \mu-\text{a.e.} \biggr\}. \nonumber
\end{align}
If $G$ is the unit ball, then $\rL_{\infty}\bigl(\mu,G\bigr)$ is also the unit ball in $\rL_{\infty}\bigl(\mu,\sY^*\bigr)$, and so, by Banach-Alaoglu Theorem \cite[Theorem 5.18]{Fol99}, it is  compact with respect to $w^*$-topology. Since $L_1\bigl(\mu,\sY\bigr)$ is separable, by \cite[Lemma 1.3.2]{HeLa03}, $\rL_{\infty}\bigl(\mu,G\bigr)$ is metrizable, and so, is also sequentially compact.

\subsection{A Particular Case}\label{particular}

Let $\sU$ be a locally compact Polish space endowed with its Borel $\sigma$-algebra $\U$. For any $g \in C_0(\sU)$, let
\begin{align}
\|g\| \triangleq \sup_{u \in \sU} |g(u)| \nonumber
\end{align}
which turns $(C_0(\sU),\|\cdot\|)$ into a separable Banach space. Let $\|\cdot\|_{TV}$ denote the total variation norm on $\M(\sU)$.

\begin{theorem}\cite[Theorem 7.17]{Fol99}
For any $\nu \in \M(\sU)$ and $g \in C_0(\sU)$, let $I_{\nu}(g) \triangleq \langle \nu , g \rangle$, where
\begin{align}
\langle \nu, g \rangle \triangleq \int_{\sU} g(u) \, \nu(du). \nonumber
\end{align}
Then the map $\nu \mapsto I_{\nu}$ is an isometric isomorphism from $\M(\sU)$ to $C_0(\sU)^*$. Hence, we can identify $C_0(\sU)^*$ with $\M(\sU)$. For any $g \in C_0(\sU)$ and $\nu \in \M(\sU)$, the duality pairing is given by
\begin{align}
\langle \nu, g \rangle = \int_{\sU} g(u) \, \nu(du).\nonumber
\end{align}
Furthermore, the norm on $\M(\sU)$ induced by this duality is total variation norm.
\end{theorem}

Let us set $\sY = C_0(\sU)$, and so, $\sY^* = \M(\sU)$. Then, we can define the following Banach spaces $L_1\bigl(\mu,C_0(\sU)\bigr)$ and $\rL_{\infty}\bigl(\mu,\M(\sU)\bigr)$ with the following norms, respectively,
\begin{align}
\|f\|_1 &= \int_{\sX} \|f(x)\| \, \mu(dx), \label{eq3} \\
\|\gamma\|_{\infty} &= \ess \sup_{x \in \sX} \|\gamma(x)\|_{TV}. \label{eq4}
\end{align}
Since $C_0(\sU)$ is separable, $L_1\bigl(\mu,C_0(\sU)\bigr)$ is also separable. By Theorem~\ref{mainduality}, the topological dual of $L_1\bigl(\mu,C_0(\sU)\bigr)$ is $\rL_{\infty}\bigl(\mu,\M(\sU)\bigr)$; that is
\begin{align}
L_1\bigl(\mu,C_0(\sU)\bigr)^* = \rL_{\infty}\bigl(\mu,\M(\sU)\bigr). \nonumber
\end{align}
For any $f \in L_1\bigl(\mu,C_0(\sU)\bigr)$ and $\gamma \in \rL_{\infty}\bigl(\mu,\M(\sU)\bigr)$, the duality pairing is given by
\begin{align}
\langle\langle \gamma, f \rangle\rangle &= \int_{\sX} \langle \gamma(x), f(x) \rangle \, \mu(dx) \nonumber \\
&= \int_{\sX} \int_{\sU} f(x)(u) \, \gamma(x)(du) \, \mu(dx). \label{eq5}
\end{align}
Hence, we can equip $\rL_{\infty}\bigl(\mu,\M(\sU)\bigr)$ with $w^*$-topology induced by $L_1\bigl(\mu,C_0(\sU)\bigr)$. Under this topology, $\gamma_{\lambda} \rightharpoonup^* \gamma$ in $\rL_{\infty}\bigl(\mu,\M(\sU)\bigr)$, if
\begin{align}
&\int_{\sX} \int_{\sU} f(x)(u)\, \gamma_{\lambda}(x)(du)\, \mu(dx) \rightarrow \int_{\sX} \int_{\sU} f(x)(u)\, \gamma(x)(du)\, \mu(dx), \nonumber
\end{align}
for all $f \in L_{1}\bigl(\mu,C_0(\sU)\bigr)$.

\begin{definition}\label{stoc-ker}
A mapping $\gamma: \sX \rightarrow \M(\sU)$ is called $\mu$-\emph{stochastic kernel} from $\sX$ to $\sU$ if, for all $D \in \U$, the mapping $\sX \ni x \mapsto \gamma(x)(D) \in \R$ is $\X / \B(\R)$-measurable and $\gamma(x) \in \P(\sU)$ $\mu$-a.e.. Let $\P_{\mu}(\sU \,|\, \sX)$ denote the set of all equivalence classes of $\mu$-stochastic kernels from $\sX$ to $\sU$. If $\gamma(x) \in \P(\sU)$ for all $x \in \sX$, it is called stochastic kernel without referring to $\mu$.
\end{definition}

The following result is the key to introduce a topology for the set of team decision rules. 

\begin{lemma}\label{kernel}
We have $\rL_{\infty}\bigl(\mu,\P(\sU)\bigr)=\P_{\mu}(\sU \,|\, \sX)$.
\end{lemma}

\begin{proof}
Let $\gamma \in \rL_{\infty}\bigl(\mu,\P(\sU)\bigr)$. Note first that the mapping $\sX \ni x \mapsto \langle \gamma(x),g \rangle \in \R$ is $\X / \B(\R)$-measurable for any continuous and bounded $g$ on $\sU$, because any such $g$ can be approximated pointwise by $\{g_n\}_{n\geq1} \subset C_0(\sU)$ satisfying $ \|g_n\| \leq \|g\|$ for all $n$. Moreover, for any closed set $F \subset \sU$, one can approximate pointwise the indicator function $1_F$ by continuous and bounded functions $h_n(u) = \max\bigl(1 - n d_{\sU}(u,F), 0 \bigr)$, where $d_{\sU}$ is the metric on $\sU$ and $d_{\sU}(u,F) = \inf_{y \in F} d_{\sU}(u,y)$. This implies that the mapping $\sX \ni x \mapsto \gamma(x)(F) \in \R$ is $\X / \B(\R)$-measurable for all closed set $F$ in $\sU$. Then the result follows by \cite[Proposition 7.25]{BeSh78}.

The reverse implication is straightforward.
\end{proof}

Let $\P_{_{\leq1}}(\sU)$ denote the set of sub-probability measures in $\M(\sU)$. Then, since $\rL_{\infty}\bigl(\mu,\P_{_{\leq1}}(\sU)\bigr)$ is proved to be closed in $\rL_{\infty}\bigl(\mu,\M(\sU)\bigr)$ with respect to $w^*$-topology and is a subset of the unit ball, it is (sequentially) compact and metrizable with respect to $w^*$-topology by Banach-Alaoglu Theorem.  
Note that $\rL_{\infty}\bigl(\mu,\P(\sU)\bigr)$ is a subset of $\rL_{\infty}\bigl(\mu,\P_{_{\leq1}}(\sU)\bigr)$, and so, we can endow $\rL_{\infty}\bigl(\mu,\P(\sU)\bigr)$ with relative $w^*$-topology. Note that $\rL_{\infty}\bigl(\mu,\P(\sU)\bigr)$ is not closed with respect to $w^*$-topology unless $\sU$ is compact. Indeed, let $\sX = \sU = \R$. Define $\gamma_n(x)(\,\cdot\,) = \delta_n(\,\cdot\,)$ and $\gamma(x)(\,\cdot\,) = 0(\,\cdot\,)$, where $\delta_a$ denotes the point mass at $a$ and $0(\,\cdot\,)$ denotes the degenerate measure on $\sU$; that is, $0(D)=0$ for all $D \in \U$. Let $f \in L_1\bigl(\mu,C_0(\sU)\bigr)$. Then we have
\begin{align}
\lim_{n\rightarrow\infty} \int_{\sX} \langle \gamma_n(x), f(x) \rangle \, \mu(dx) &= \lim_{n\rightarrow\infty} \int_{\sX} f(x)(n) \, \mu(dx) \nonumber \\
&= \int_{\sX} \lim_{n\rightarrow\infty} f(x)(n) \, \mu(dx) \text{ (as $\|f(x)\|$ is $\mu$-integrable)} \nonumber \\
&= 0 \text{ (as $f(y) \in C_0(\sU)$)}. \nonumber
\end{align}
Hence, $\gamma_n \rightharpoonup^* \gamma$. But, $\gamma \notin \rL_{\infty}\bigl(\mu,\P(\sU)\bigr)$, and so, $\rL_{\infty}\bigl(\mu,\P(\sU)\bigr)$ is not closed in $\rL_{\infty}\bigl(\mu,\P_{_{\leq1}}(\sU)\bigr)$. Therefore, $\rL_{\infty}\bigl(\mu,\P(\sU)\bigr)$ is relatively (sequentially) compact with respect to $w^*$-topology. 

\subsection{The $w^*$-topology and Young narrow topology}

In this section, we compare $w^*$-topology on $\mu$-stochastic kernels $\rL_{\infty}\bigl(\mu,\P(\sU)\bigr)$ with Young narrow topology \cite{Val94}. We also refer the reader to the excellent recent paper \cite{Yuk23} for comparison of different topologies, including $w^*$-topology and Young narrow topology, on control policies, which are used to prove the continuous dependence of invariant measures on control policy.

A first relevant result toward this direction is the identification of $L_1\bigl(\mu,C_0(\sU)\bigr)$ as a subset of the set of  Caratheodory functions. 

\begin{definition}
A measurable function $h:\sX\times\sU\rightarrow\R$ is called a Caratheodory function if it is continuous in $u$ for all $x \in \sX$. Let $\Car(\sX\times\sU)$ denote the set of all Caratheodory functions. Let $\Car_0(\sX\times\sU)$ denote the set of all $h \in \Car(\sX\times\sU)$ such that $h(x,\,\cdot\,)$ vanishes at infinity and let $\Car_b(\sX\times\sU)$ denote the set of all bounded Caratheodory functions. 
\end{definition} 

The following result states that the set of strongly measurable functions is a subset of the set of Caratheodory functions.  

\begin{lemma}\label{Car}
Let $f:\sX \rightarrow C_0(\sU)$ be a strongly measurable function and define $h_f:\sX\times\sU\rightarrow \R$ as $h_f(x,u) \triangleq f(x)(u)$. Then, $h_f \in \Car_0(\sX\times\sU)$. Conversely, let $h \in \Car_0(\sX\times\sU)$ and define $f_h:\sX\rightarrow C_0(\sU)$ as $f_h(x) = h(x,\,\cdot\,)$. Then, $f_h$ is strongly measurable. 
\end{lemma}

\begin{proof}
It is straightforward to prove that the forward statement is true for any simple function $f:\sX\rightarrow C_0(\sU)$.  Since any strongly measurable function can be approximated via simple functions by definition, the forward statement is also true for any strongly measurable function. 

Conversely, let $h\in\Car_0(\sX\times\sU)$ and define $f_h:\sX\rightarrow C_0(\sU)$ as $f_h(x) = h(x,\,\cdot\,)$. Then, for any $\gamma \in \M(\sU)$, the function $\sX \ni x \mapsto \int_{\sU} f_h(x)(u) \, \gamma(du) \in \R$ is measurable. Hence, $f_h$ is strongly measurable by \cite[Theorem 4.2.4-(c)]{PaWi18}.  
\end{proof}

Therefore, any element $f$ of some equivalence class in  $L_1\bigl(\mu,C_0(\sU)\bigr)$ is an element of $\Car_0(\sX\times\sU)$ satisfying 
\begin{align}\label{cccv}
\int_{\sX} \ess \sup_{u \in \sU} |f(x)(u)| \, \mu(dx) < \infty.
\end{align}
The converse is also true; that is, if $h \in \Car_0(\sX\times\sU)$ satisfying (\ref{cccv}), then $h$ is an  element of some equivalence class in $L_1\bigl(\mu,C_0(\sU)\bigr)$.

Let us now define \emph{Young narrow topology} on the set of $\mu$-sub-stochastic kernels $\rL_{\infty}\bigl(\mu,\P_{_{\leq1}}(\sU)\bigr)$ \cite[Definition 4.7.11]{PaWi18}. Young narrow topology on $\rL_{\infty}\bigl(\mu,\P_{_{\leq1}}(\sU)\bigr)$ is the smallest topology for which the mapping
\begin{align}
\rL_{\infty}\bigl(\mu,\P_{_{\leq1}}(\sU)\bigr) \ni \gamma \mapsto \int_{\sX} \int_{\sU} f(x,u) \, \gamma(x)(du) \, \mu(dx) \in \R, \nonumber
\end{align}
is continuous for all equivalence classes $f$ of bounded Caratheodory functions. In this case, the equivalence relation is the same as the equivalence relation in $L_1(\mu,C_0(\sU))$; that is, $f \sim \tilde{f}$ in $\Car_b(\sX\times\sU)$ if $f(x,\,\cdot\,)=\tilde{f}(x,\,\cdot\,)$ $\mu$-a.e.. Note that in $w^*$-topology, we consider equivalence classes in $L_1(\mu,C_0(\sU))$, where every function in these equivalence classes is an element of $\Car_0(\sX\times\sU)$ and satisfying (\ref{cccv}). Hence, the relation between $w^*$-topology and Young narrow topology is similar to the vague topology and weak topology on measures.

In Young narrow topology, one can prove that $\rL_{\infty}\bigl(\mu,\P(\sU)\bigr)$ is a closed (as opposed to $w^*$-topology) and metrizable subset of $\rL_{\infty}\bigl(\mu,\P_{_{\leq1}}(\sU)\bigr)$ \cite[Proposition 4.7.14]{PaWi18}. However, $\rL_{\infty}\bigl(\mu,\P(\sU)\bigr)$  is not relatively (sequentially) compact in this topology unless $\sU$ is compact. Moreover, we also do not have nice duality structure as we have in $w^*$-topology. These last two properties are extremely important for proving the existence of optimal team decision rule. These are indeed the motivations for working with $w^*$-topology instead of Young narrow topology on $\rL_{\infty}\bigl(\mu,\P(\sU)\bigr)$ in this paper. 

It is interesting to note that by \cite[Theorem 3]{Val94}, $w^*$-topology and Young narrow topology are topologically equivalent on $\rL_{\infty}\bigl(\mu,\P(\sU)\bigr)$. However, they are induced by two different topologies         
on $\rL_{\infty}\bigl(\mu,\P_{_{\leq1}}(\sU)\bigr)$. This is indeed the reason for the above distinctions between $w^*$ and Young narrow topology on $\rL_{\infty}\bigl(\mu,\P(\sU)\bigr)$ as completeness is not a topological property.

\section{Centralized Reduction of Static Team Problem}\label{staticteams}

In this section, using the $w^*$-topology on stochastic kernels, we establish the centralized reduction of static team problem via common information approach. Let us set
\begin{align}
\sY = \prod_{i=1}^N \sY_i \text{ } \text{ and } \text{ } \sU = \prod_{i=1}^N \sU_i. \nonumber
\end{align}
Let $\by$ and $\bu$ denote the elements of these sets, respectively.

%
%

Recall the cost function $c: \sX \times \sX_0 \times \prod_i \sY_i \times \prod_i \sU_i \to [0,\infty)$. We define the function $\tilde{c}:\sX \times \sX_0 \times \prod_i \sY_i \times \prod_i \M(\sU_i)  \rightarrow \R_+$ as
\begin{align}
&\tilde{c}(x,x_0,\by,\nu_1,\nu_2,\ldots,\nu_N) \triangleq \int_{\sU} c(x,x_0,\by,u_1,u_2,\ldots,u_N) \, \prod_{i=1}^N \nu_i(du_i). \nonumber
\end{align}
Then, the static team problem is equivalent to the following optimization problem:
\begin{align}
\mathbf{(S)} \inf_{\substack{\gamma_i \in \Gamma_i \\ i=1,\ldots,N}} \cE \biggl[ \tilde{c}\biggl(x,x_0,\by,\gamma_1(x_0,y_1),\ldots,\gamma_N(x_0,y_N)\biggr) \biggr]. \nonumber
\end{align}

We now transform $\mathbf{(S)}$ into a centralized stochastic control problem via common information approach.

For each $i=1,\ldots,N$, let us set
\begin{align}
\Lambda_i \triangleq \rL_{\infty}\bigl(\mu_i,\P(\sU_i)\bigr). \nonumber
\end{align}
We endow $\Lambda_i$ with $w^*$-topology induced by the functions $L_1\bigl(\mu_i,C_0(\sU_i)\bigr)$. Under this $w^*$-topology, $\Lambda_i$ is metrizable and relatively (sequentially) compact. Therefore, it is also separable. Borel $\sigma$-algebra generated by this topology is denoted by $\B\bigl(\Lambda_i\bigr)$.
Let
\begin{align}
 \Lambda \triangleq \Lambda_1 \times \Lambda_2 \times \ldots \times \Lambda_N. \nonumber
\end{align}
We endow $\Lambda$ with the product topology, and so,  Borel $\sigma$-algebra $\B(\Lambda)$ is
\begin{align}
\B(\Lambda) = \B(\Lambda_1) \otimes \B(\Lambda_2) \otimes \ldots \otimes \B(\Lambda_N). \nonumber
\end{align}
Using these definitions, we introduce the following class of functions:
\begin{align}
\A \triangleq \biggl\{ \lambda: \sX_0 \rightarrow \Lambda; \text{ }  \lambda \text{ is } \X_0/\B(\Lambda)-\text{measurable} \biggr\}. \nonumber
\end{align}
In centralized reduction, $\Lambda$ will be the set of possible 
actions for the centralized agent and $\A$ is the set of corresponding policies.

Note that given any $\lambda = (\lambda_1,\lambda_2,\ldots,\lambda_N) \in \A$, one can view $\lambda_i$ as a $\mu_0\otimes\mu_i$-stochastic kernel from $\sX_0 \times \sY_i$ to $\sU_i$, since $\lambda_i(x_0) \in \Lambda_i \triangleq \rL_{\infty}\bigl(\mu_i,\P(\sU_i)\bigr)$ for any $x_0 \in \sX_0$.
Indeed, for any $\lambda_i$, let us define
$$
\gamma_{\lambda_i}(x_0,y_i) \triangleq \lambda_i(x_0)(y_i). 
$$
Analogous to Lemma~\ref{kernel}, we can prove the following result.

\begin{lemma}\label{kernel2}
Given any $\lambda = (\lambda_1,\lambda_2,\ldots,\lambda_N) \in \A$, $\gamma_{\lambda_i}$ is a $\mu_0\otimes\mu_i$-stochastic kernel from $\sX_0 \times \sY_i$ to $\sU_i$ for any $i=1,\ldots,N$. Conversely, if $\hat{\gamma}_i: \sX_0 \times \sY_i \rightarrow \M(\sU_i)$ is a $\mu_0\otimes\mu_i$-stochastic kernel for each $i=1,\ldots,N$, then there exist stochastic kernels $\gamma_i$ such that $\gamma_i = \hat{\gamma}_i$ $\mu_0 \otimes \mu_i$-a.e., $\gamma_i(x_0,\cdot) \in \Lambda_i$ for all $x_0 \in \sX_0$ and the following map
\begin{align}
\sX_0 \ni x_0 \mapsto \bigl(\gamma_1(x_0,\cdot),\gamma_2(x_0,\cdot),\ldots,\gamma_N(x_0,\cdot)\bigr) \in \Lambda \label{aux1}
\end{align}
is in $\A$.
\end{lemma}

\begin{proof}
For any $i=1,\ldots,N$, let $\G_i$ be the smallest $\sigma$-algebra on $\Lambda_i$ that makes functions of the form
\begin{align}
\Lambda_i \ni \gamma_i \mapsto \langle\langle \gamma_i,f \rangle\rangle \in \R \phantom{xxx} \biggl(f \in L_{1}\bigl(\mu_i,C_0(\sU_i)\bigr)\biggr) \label{eq8}
\end{align}
measurable. Note that $w^*$-topology on $\Lambda_i$ is the smallest topology that makes functions given in (\ref{eq8}) continuous. Since the indicator function of any closed set in $\Lambda_i$ can be approximated pointwise by continuous and bounded functions and since $\B(\Lambda_i)$ is the smallest $\sigma$-algebra that contains closed sets, we can conclude that $\B(\Lambda_i)$ is the smallest $\sigma$-algebra that makes functions given in (\ref{eq8}) measurable. Hence, $\B(\Lambda_i) = \G_i$.

We first prove the converse. Let $\hat{\gamma}_i: \sX_0 \times \sY_i \rightarrow \M(\sU_i)$ is a $\mu_0\otimes\mu_i$-stochastic kernel for each $i=1,\ldots,N$. Then, for each $i$, there exists a stochastic kernel $\gamma_i$ in the equivalence class of $\hat{\gamma}_i$ such that $\gamma_i = \hat{\gamma}_i$ $\mu_0 \otimes \mu_i$-a.e. and $\gamma_i(x_0,y_i) \in \P(\sU_i)$ for all $(x_0,y_i) \in \sX_0 \times \sY_i$. Hence, $\gamma_i(x_0,\cdot) \in \Lambda_i$ for all $x_0 \in \sX_0$ by Lemma~\ref{kernel}. Note first that the mapping
\begin{align}
\sX_0 \ni x_0 \mapsto \int_{\sY_i \times \sU_i} g(u_i) 1_{A}(y_i) \gamma_i(x_0,y_i)(du_i) \, \mu_i(dy_i) \in \R \nonumber
\end{align}
is measurable for any Borel $A \subset \sY_i$ and $g \in C_0(\sU_i)$. Since any $f \in L_{1}\bigl(\mu_i,C_0(\sU_i)\bigr)$ can be approximated by simple functions $\sum_{k=1}^m 1_{A_k}(x_i) g_k(u_i)$,
the following is also measurable
\begin{align}
\sX_0 \ni x_0 \mapsto \langle\langle \gamma_i(x_0,\cdot),f \rangle\rangle \in \R \nonumber
\end{align}
for any $f \in L_{1}\bigl(\mu_i,C_0(\sU_i)\bigr)$. But this is the same as $\X_0/\G_i$ measurability of $\gamma_i(x_0,\cdot)$. Since $\B(\Lambda_i) = \G_i$ for all $i$, the mapping in (\ref{aux1}) is $\X_0/\B(\Lambda)$-measurable. This completes the proof of converse part.

Let $\lambda = (\lambda_1,\ldots,\lambda_N) \in \A$. Then $\lambda_i$ is $\X_0/\B(\Lambda_i)$-measurable for each $i=1,\ldots,N$. Since $\B(\Lambda_i) = \G_i$, this implies that the following is also measurable
\begin{align}
\sX_0 \ni x_0 \mapsto \langle\langle \lambda_i(x_0),f \rangle\rangle \in \R \nonumber
\end{align}
for any $f \in L_{1}\bigl(\mu_i,C_0(\sU_i)\bigr)$. Define $\sM_i: \sX_0 \ni x_0 \mapsto \lambda_i(x_0)(y_i)(du_i) \otimes \mu_i(dy_i) \in \P(\sY_i \times \sU_i)$. For any $g \in C_0(\sU_i)$ and Borel $A \subset \sY_i$, the following map is measurable from $\sX_0$ to $\R$:
\begin{align}
&\int_{\sY_i}  \int_{\sU_i} 1_{A}(y_i) g(u_i) \lambda_i(x_0)(y_i)(du_i) \, \mu_i(dy_i) \triangleq \int_{\sY_i \times \sU_i} 1_{A}(y_i) g(u_i) \, \sM_i(x_0)(dy_i,du_i). \nonumber
\end{align}
This is also true if we replace $g$ with bounded and continuous function on $\sU_i$. Using this, we can also replace $g$ with an indicator function of any closed set $F$ in $\sU_i$. Since sets of the form $A\times F$, where $F \subset \sU_i$ closed, generates the product Borel $\sigma$-algebra on $\sY_i \times \sU_i$, the mapping $\sM_i$ is a stochastic kernel from $\sX_0$ to $\sY_i \times \sU_i$ \cite[Proposition 7.25]{BeSh78}. Then, by \cite[Proposition 7.27]{BeSh78}, there exists a stochastic kernel $q_i: \sX_0 \times \sY_i \rightarrow \P(\sU_i)$ such that
\begin{align}
\lambda_i(x_0)(y_i)(du_i) \otimes  \mu_i(dy_i) &\triangleq \sM_i(x_0)(dy_i,du_i) \nonumber \\
&= q_i(x_0,y_i)(du_i) \otimes  \mu_i(dy_i). \nonumber
\end{align}
Therefore, for any $x_0$, $\lambda_i(x_0)(y_i)(du_i) = q_i(x_0,y_i)(du_i)$ $\mu_i$-a.e.. Hence, $\gamma_{\lambda_i}(x_0,y_i) \triangleq \lambda_i(x_0)(y_i)$ is a $\mu_0\otimes\mu_i$-stochastic kernel. 
\end{proof}

Lemma~\ref{kernel2} states that ${\bf \Gamma}$ is equivalent to the set $\A$; that is,
$$
{\bf \Gamma} = \A. 
$$
Let us define the function $L: \sX \times \sX_0 \times \Lambda \rightarrow \R_+$ as
\begin{align}
&L(x,x_0,\lambda) \triangleq \int_{\sY} \tilde{c}(x,x_0,\by,\lambda_1(x_0)(y_1),\ldots,\lambda_N(x_0)(y_N)) \, \prod_{i=1}^N q_i(y_i,x) \, \mu_i(dy_i), \nonumber
\end{align}
where $\lambda \triangleq (\lambda_1,\ldots,\lambda_N)$.

\begin{lemma}\label{measurable}
The function $L$ is $\X \otimes \X_0 \otimes \B(\Lambda) /\B(\R)$-measurable.
\end{lemma}

\begin{proof}
Note that
\begin{align}
L(x,x_0,\lambda) &= \int_{\sY \times \sU} c(x,x_0,\by,\bu) \, \prod_{i=1}^N q(y_i|x) \, \prod_{i=1}^N \lambda_i(y_i)(du_i) \otimes \mu_i(dy_i). \nonumber
\end{align}
One can prove that the following function is $\X \otimes \X_0 \otimes \B(\Lambda) /\B(\R)$-measurable:
\begin{align}
\sX \times \sX_0 \times \Lambda \ni (x,x_0,\gamma) &\mapsto \int_{\sY \times \sU} 1_{A}(\bu) \, 1_B(\by) \, 1_C(x_0) \, 1_D(x) \, \nonumber \\
&\phantom{xxxxxxxxxxxxx}\prod_{i=1}^N \gamma_i(y_i)(du_i) \otimes \mu_i(dy_i) \in \R, \label{eq10}
\end{align}
where $A \in \U$, $B = \prod_{i=1}^N B_i \in \X$, $C \in \X_0$, and $D \in \X$. Indeed, this can be established by first proving the measurability of the following function:
\begin{align}
\sX \times \sX_0 \times \Lambda \ni (x,x_0,\gamma) &\mapsto \int_{\sY \times \sU} \prod_{i=1}^N g_i(u_i) \, 1_B(\by) \, 1_C(x_0) \, 1_D(x) \, \nonumber \\
&\phantom{xxxxxxxxxxx}\prod_{i=1}^N \gamma_i(x_i)(du_i) \otimes \mu_i(dy_i) \in \R, \label{eq9}
\end{align}
where $g_i \in C_c(\sU_i)$ for all $i$. Then by Stone-Weiestrass Theorem \cite[Lemma 6.1]{Lan93}, any continuous and compactly supported $g$ on $\sU$ can be uniformly approximated by functions of the form
\begin{align}
\sum_{j=1}^m \prod_{i=1}^N g_{j,i}(u_i), \nonumber
\end{align}
where $g_{j,i} \in C_c(\sU_i)$ for all $i,j$. Hence, if we replace $\prod_{i=1}^N g_i(u_i)$ with any continuous and compactly supported $g$ on $\sU$ in (\ref{eq9}), the function is still measurable. This is also true if $g$ is nonnegative, continuous and bounded, since, for any such $g$, we can find a sequence of continuous and compactly supported functions $\{g_n\}$ such that $0 \leq g_1 \leq g_2 \leq \ldots \leq g_n \leq \ldots \leq g$ and $g_n$ converges to $g$ pointwise (see the proof of \cite[Proposition 1.4.18]{HeLa03}). Finally, for any closed set $F \subset \sU$, one can approximate pointwise the indicator function $1_F$ by continuous and bounded functions $h_n(u) = \max\bigl(1 - n d_{\sU}(u,C), 0 \bigr)$, where $d_{\sU}$ is the product metric on $\sU$ and $d_{\sU}(u,C) = \inf_{y \in C} d_{\sU}(u,y)$. Hence, (\ref{eq10}) is measurable if $A$ is closed. Since $\sU$ is endowed with the Borel $\sigma$-algebra $\U$ (the smallest $\sigma$-algebra that contains the closed sets), measurability is still true if $A$ is in $\U$. Therefore, (\ref{eq10}) is measurable. Then, this implies the measurability of (\ref{eq10}) when $1_{A}(\bu) 1_B(\by) 1_C(x_0) 1_D(x)$ is replaced by $1_G$ where $G \in \U \otimes \Y \otimes \X_0 \otimes \X$. Hence, measurability holds for simple functions, and therefore, for any measurable function such as $c(x,x_0,\by,\bu) \prod_{i=1}^N q(y_i|x)$.
\end{proof}

In the view of previous results, the following static centralized control problem is equivalent to the static team problem $\mathbf{(S)}$:
\begin{align}
\mathbf{(SC)} \inf_{\lambda \in \A} \cE \biggl[ L\biggl(x,x_0,\lambda(x_0)\biggr) \biggr]. \nonumber
\end{align}
Recall that
\begin{align}
\A \triangleq \biggl\{ \lambda: \sX_0 \rightarrow \Lambda; \text{ }  \lambda \text{ is } \X_0/\B(\Lambda)-\text{measurable} \biggr\}. \nonumber
\end{align}
Hence, one can interpret this static centralized control problem as follows: $x$ represents the state of the system and $x_0$ represents the observation of the decision maker. Depending on the observation $x_0$, decision maker chooses its strategy from the control space $\Lambda$ to minimize the expectation of the cost function $L: \sX \times \sX_0 \times \Lambda  \rightarrow \R_+$. Therefore, if $\lambda^* \in \A$ satisfies the following
\begin{align}
&\inf_{\lambda \in \Lambda} \cE \biggl[ L\biggl(x,x_0,\lambda\biggr) \bigg| x_0 \biggr] = \cE \biggl[ L\biggl(x,x_0,\lambda^*(x_0)\biggr) \bigg| x_0 \biggr] \nonumber
\end{align}
for $\mu_0$-a.e. $x_0 \in \sX_0$, then $\lambda^*$ is an optimal policy. The existence of such policy can be established using measurable selection theorems. To be able to apply such theorems, one needs to establish some properties of the set $\Lambda$ and of the following function:
\begin{align}
M(x_0,\lambda): \sX_0 \times \Lambda \ni (x_0,\lambda) \mapsto \cE \biggl[ L\biggl(x,x_0,\lambda\biggr) \bigg| x_0 \biggr] \in \R. \label{eq11}
\end{align}
For instance, the optimal selector exists if $M(x_0,\,\cdot\,)$ is inf-compact for all $x_0 \in \sX$ \cite[Theorem 2.2]{FeKa21}.

\begin{remark}
Note that, in addition to the existence of team optimal policies, we can also use the centralized reduction $\mathbf{(SC)}$ to establish approximation results for such problems. Indeed, since $\Lambda$ is relatively sequentially compact, it can be approximated by a finite set with arbitrary precision. Moreover, as $\sX_0$ is a locally compact, separable, and complete metric space, we can also approximate $\sX_0$ with a finite set. Consequently, the overall approximate problem becomes finite and can be solved more easily. In this context, we can apply the approximation method introduced in \cite{SaYuLi17} to discretize the action space $\Lambda$.
\end{remark}

\begin{theorem}\label{main1}
$M(x_0,\,\cdot\,)$ is lower semi-continuous for any $x_0 \in \sX_0$. 
\end{theorem}

\begin{proof}
Fix any $x_0 \in \sX_0$. Recall that $C_c(\sX \times \sX_0 \times \sY \times \sU)$ denotes the set of real continuous functions on $\sX \times \sX_0 \times \sY \times \sU$ with compact support. For any $g \in C_c(\sX \times \sX_0 \times \sY \times \sU)$, we define
\begin{align}
M_g(x_0,\lambda) = \int_{\sX\times\sY \times \sU} g(x,x_0,{\bf y},{\bf u}) \, \prod_{i=1}^N \lambda_i(y_i)(du_i) \,q_i(y_i,x) \mu_i(dy_i) \, \rP(dx|x_0). \nonumber
\end{align}
Note that 
\begin{align}
&M(x_0,\lambda) = \int_{\sX\times\sY \times \sU} c(x,x_0,{\bf y},{\bf u}) \, \prod_{i=1}^N \lambda_i(y_i)(du_i) \,q_i(y_i,x) \mu_i(dy_i) \, \rP(dx|x_0). \nonumber
\end{align}
As $\sX \times \sX_0 \times \sY \times \sU$ is locally compact, one can find a sequence of $\{g_m\} \subset C_c(\sX \times \sX_0 \times \sY \times \sU)$ such that $0\leq g_1 \leq g_2 \leq \ldots \leq g_m \leq \ldots \leq c$ and $g_m \rightarrow c$ pointwise (see the proof of \cite[Proposition 1.4.18]{HeLa03}). Therefore, it is sufficient to prove that $M_g(x_0,\,\cdot\,)$ is lower semi-continuous (what is more, continuous) for any $g \in C_c(\sX \times \sX_0 \times \sY \times \sU)$.

Fix any $g \in C_c(\sX \times \sX_0 \times \sY \times \sU)$. Then by Stone-Weierstrass Theorem \cite[Lemma 6.1]{Lan93}, $g$ can be uniformly approximated by functions of the form
\begin{align}
\sum_{j=1}^k r_j t_j \prod_{i=1}^N f_{j,i} g_{j,i}, \nonumber
\end{align}
where $r_j \in C_c(\sX)$, $t_j \in C_c(\sX_0)$, $f_{j,i} \in C_c(\sY_i)$, and $g_{j,i} \in C_c(\sU_i)$ for each $j=1,\ldots,k$ and $i=1,\ldots,N$. This implies that it is sufficient to prove the continuity of $M_g$ for functions of the form $r \, t \, \prod_{i=1}^N f_{i} g_{i}$, where $r \in C_c(\sX)$, $t \in C_c(\sX_0)$, $f_{i} \in C_c(\sY_i)$, and $g_{i} \in C_c(\sU_i)$ for $i=1,\ldots,N$.  Therefore, in the sequel, we assume that $g = r \, t \, \prod_{i=1}^N f_{i} g_{i}$.

Let $\lambda^n \rightharpoonup^* \lambda$. Then 
\begin{align}
&M_g(x_0,\lambda^n) = t(x_0) \int_{\sX \times \sY \times \sU} \hspace{-16pt} r(x) \, \prod_{i=1}^N f_i(y_i) g_i(u_i) q_i(y_i,x) \, \prod_{i=1}^N \lambda_i^n(y_i)(du_i) \, \mu_i(dy_i) \, \rP(dx|x_0). \nonumber
\end{align}
Define 
\begin{align}
&J_g(\lambda^n) = \int_{\sX \times \sY \times \sU} r(x) \, \prod_{i=1} f_i(y_i) g_i(u_i) q_i(y_i,x) \, \prod_{i=1}^N \lambda_i^n(y_i)(du_i) \, \mu_i(dy_i) \, \rP(dx|x_0). \nonumber
\end{align}
Then we have
\begin{align}
\bigl| J_g(\lambda^n) - J_g(\lambda) \bigr| &\leq  \bigl| J_g(\lambda_1^n,\ldots,\lambda_N^n) - J_g(\lambda_1,\lambda_2^n,\ldots,,\lambda_N^n) \bigr| \nonumber \\
&+  \bigl| J_g(\lambda_1,\lambda_2^n,\ldots,\lambda_N^n) - J_g(\lambda_1,\lambda_2,\lambda_3^n,\ldots,\lambda_N^n) \bigr| \nonumber \\
&\phantom{x}\vdots \nonumber \\
&+  \bigl| J_g(\lambda_1,\ldots,\lambda_{N-1}, \lambda_N^n) - J_g(\lambda_1,\ldots,\lambda_N) \bigr| \nonumber \\
&\eqqcolon \sum_{j=1}^N l_j^{(n)}. \nonumber
\end{align}
Let us consider the $j^{th}$ term in the above expression. Define the probability measure $T_{-j}$ on $\sX \times \sY_{-j} \times \sU_{-j}$ and real function $g_{-j}$ on $\sX \times \sY_{-j} \times \sU_{-j}$ as follows:
\small
\begin{align}
&T_{-j} = \biggl( \prod_{i=1}^{j-1}  \lambda_i(du_i|y_i) q_i(y_i,x) \mu_i(dy_i) \biggr) \times \biggl( \prod_{i=j+1}^{N} \hspace{-5pt} \gamma_i^{(n)}(du_i|y_i) q_i(y_i,x) \mu_i(dy_i) \biggr) \rP(dx|x_0) \nonumber
&\intertext{and}
&g_{-j} \coloneqq r \prod_{i \neq j} f_i g_i. \nonumber
\end{align}
\normalsize
Then the $j^{th}$ term can be written as
\begin{align}
l_j^{(n)} &= \biggl| \int g_{-j} \biggl( \int f_j g_j q_j d\lambda_j^{(n)}\otimes\mu_j \biggr) dT_{-j} - \int g_{-j} \biggl( \int f_j g_j q_j d\lambda_j\otimes\mu_j \biggr) dT_{-j} \biggr|. \nonumber
\end{align}

Define, for each $x \in \sX$, the function
\begin{align}
b_x(y_j,u_j) \coloneqq f_j(y_j) g_j(u_j) q_j(y_j,x). \nonumber
\end{align}
It is obvious that $b_x \in L_1(\mu_j,C_0(\sU_j))$ for all $x \in \sX$. We will prove that the set $\{b_x\}_{x \in K} \subset L_1(\mu_j,C_0(\sU_j))$ is totally bounded, where $K \subset \sX$ is the compact support of the function $r$ on $\sX$. Indeed, let $x,\tilde{x} \in K$. Then
\begin{align}
&\|b_x - b_{\tilde{x}}\|_1 \coloneqq \int_{\sY_j} \ess \sup_{u_j \in \sU_j} \bigl| f_j(y_j) g_j(u_j) q_j(y_j,x) - f_j(y_j) g_j(u_j) q_j(y_j,\tilde{x}) | \mu_j(dy_j) \nonumber \\
&\leq \|f_j\|_{\infty} \|g_j\|_{\infty} \int_{\sY_j} \bigl| q_j(y_j,x) - q_j(y_j,\tilde{x}) | \mu_j(dy_j) \nonumber \\
&= \|f_j\|_{\infty} \|g_j\|_{\infty} \bigl\|W_j(\,\cdot\,|x) - W_j(\,\cdot\,|\tilde{x})\bigr\|_{TV}. \label{unif}
\end{align}
Since $W_j$ is assumed to be continuous with respect to the total variation norm, the set $\{b_x\}_{x \in K}$ is totally bounded; that is, for any $\varepsilon > 0$, there exists a finite number of points $x_1,\ldots,x_n \in K$ such that
\begin{align}
\{  b_x  \}_{x \in K} \subset \bigcup_{i=1}^n B_1(b_{x_i},\epsilon), \nonumber
\end{align}
where $B_1(b_{x},\varepsilon) \coloneqq \{b \in L_1(\mu_j,C_0(\sU_j)): \|b-b_x\|_1 \leq \varepsilon\}$. Indeed, fix any $\varepsilon > 0$. Note first that the observation kernel $W_j: K \rightarrow \P(\sY_j)$ is uniformly continuous since $K$ is compact. Hence for any $\epsilon > 0$, one can find $\delta > 0$ such that if $d_{\sX}(x,y) < \delta$, then
$\|W_j(\,\cdot\,|x) - W_j(\,\cdot\,|y)\|_{TV} < \epsilon$. For this $\delta > 0 $, since $K$ is compact, one can find a finite number of points $x_1,\ldots,x_n \in K$ such that
\begin{align}
K \subset \bigcup_{i=1}^n B(x_i,\delta),  \nonumber
\end{align}
where $B(x,\delta) \coloneqq \{y \in \sX: d_{\sX}(x,y) \leq \delta\}$. But this implies that
\begin{align}
\{  b_x  \}_{x \in K} \subset \bigcup_{i=1}^n B_1\bigl(b_{x_i},\epsilon \|f_j\|_{\infty} \|g_j\|_{\infty} \bigr) \nonumber
\end{align}
because if $b_x$ is some element in the set  $\{  b_x  \}_{x \in K}$, then $x$ is in $B(x_i,\delta)$ for some $i$; that is, $d_{\sX}(x,x_i) < \delta$. This implies from uniform continuity that $\|W_j(\,\cdot\,|x) - W_j(\,\cdot\,|x_i)\|_{TV} < \epsilon$, and so,  by (\ref{unif}), we have $\|b_{x_i} - b_x\|_1 < \epsilon \|f_j\|_{\infty} \|g_j\|_{\infty}$.
By choosing $\epsilon = \varepsilon/(\|f_j\|_{\infty} \|g_j\|_{\infty})$, we complete the proof of the assertion.

Using the total boundedness of the set $\{b_x\}_{x \in K}$, we prove the following:
\begin{align}
\lim_{n\rightarrow\infty} \sup_{x \in K} \bigl| \langle\langle \lambda_j^{(n)},b_x \rangle \rangle - \langle\langle \lambda_j,b_x \rangle\rangle \bigr| = 0. \label{des1}
\end{align}
Suppose (\ref{des1}) is not true. Then there exists a sub-sequence $\{\lambda_j^{n_k}\}$ of $\{\lambda_j^{n}\}$ such that, for all $k$, we have
\begin{align}
\sup_{x \in K} \bigl| \langle\langle  \lambda_j^{(n_k)},b_x \rangle\rangle - \langle\langle \lambda_j,b_x \rangle\rangle \bigr| > 0
\end{align}
Suppose $\{\varepsilon_k\}$ be a sequence of positive real numbers converging to zero. For each $k$, let $x_k \in K$ be such that
\begin{align}
\bigl| \langle\langle \lambda_j^{(n_k)},b_{x_k} \rangle\rangle - \langle\langle \lambda_j,b_{x_k} \rangle\rangle \bigr| &> \sup_{x \in K} \bigl| \langle\langle \lambda_j^{(n_k)},b_x \rangle\rangle - \langle\langle \lambda_j,b_x \rangle\rangle \bigr| - \varepsilon_k \nonumber \\
&> 0. \label{cont}
\end{align}
Since $\{b_{x_{k}}\}$ is totally bounded, there exists a subsequence $\{b_{x_{k_l}}\}$ such that
\begin{align}
b_{x_{k_l}} \rightarrow b \in L_1(\mu_j,C_0(\sU_j)) \text{ in $L_1$-norm.} \nonumber
\end{align}
Then, we have
\begin{align}
\langle\langle \lambda_j^{(n_{k_l})},b_{x_{k_l}} \rangle\rangle \rightarrow \langle\langle \lambda_j,b \rangle\rangle \label{important}
\end{align}
as $\lambda_j^{(n_{k_l})} \rightharpoonup^* \lambda_j$ by Uniform Bounded Principle. In addition, we also have $$\langle\langle \lambda_j,b_{x_{k_l}} \rangle\rangle \rightarrow \langle\langle \lambda_j, b \rangle\rangle.$$ Hence,
\begin{align}
\lim_{l\rightarrow\infty} \bigl| \langle\langle \lambda_j^{(n_{k_l})},b_{x_{k_l}} \rangle\rangle - \langle\langle \lambda_j,b_{x_{k_l}} \rangle\rangle \bigr| = 0. \nonumber
\end{align}
This contradicts with (\ref{cont}), and so, (\ref{des1}) is true.

Note that we have
\begin{align}
\int f_j g_j q_j d\lambda_j^{(n)}\otimes\mu_j - \int f_j g_j q_j d&\lambda_j\otimes\mu_j = \langle\langle  \lambda_j^{(n)},b_x \rangle\rangle - \langle\langle  \lambda_j,b_x \rangle\rangle. \nonumber
\end{align}
Therefore, we can bound $l_j^{(n)}$ as
\begin{align}
l_j^{(n)} &\leq \int g_{-j} \bigl| \langle\langle  \lambda_j^{(n)},b_x \rangle\rangle - \langle\langle  \lambda_j,b_x \rangle\rangle \bigr| \, dT_{-j} \nonumber \\
&\leq \|g_{-j}\|_{\infty} \sup_{x\in K} \bigl| \langle\langle \lambda_j^{(n)},b_x \rangle\rangle - \langle\langle \lambda_j,b_x \rangle\rangle \bigr|. \nonumber
\end{align}
Note that the last term converges to zero as $n\rightarrow\infty$ by (\ref{des1}). Since $j$ is arbitrary, $l_j^{(n)} \rightarrow 0$ as $n\rightarrow\infty$ for all $j=1,\ldots,N$. This implies that $J_g(\lambda^n) \rightarrow J_g(\lambda)$, which proves that $M_g(x_0,\,\cdot\,)$ is continuous.
\end{proof}

\begin{lemma}\label{prf1}
Let $\C \subset \Lambda$ be such that the following family of probability measures are tight
$$
\P_{\C} \triangleq \left\{ \prod_{i=1}^N \lambda(y_i)(du_i) \otimes \mu_i(dy_i): (\lambda_1,\ldots,\lambda_N) \in \C \right\}.
$$
Then, the closure of $\C$ with respect to $w^*$-topology is in $\Lambda$.
\end{lemma}

\begin{proof}
Let $(\lambda_1^n,\ldots,\lambda_N^n) \rightharpoonup^* (\lambda_1,\ldots,\lambda_N)$ where $(\lambda_1^n,\ldots,\lambda_N^n) \in \C$ for all $n$.
For any $i=1,\ldots,N$, we prove that $\lambda_i(y_i) \in \P(\sU_i)$ $\mu_i$-a.e., which completes the proof. 

For any $n$, define 
\begin{align}
\nu^n(d\by,d\bu) &\triangleq \prod_{i=1}^N \lambda^n(y_i)(du_i) \otimes \mu_i(dy_i) \nonumber \\
\nu(d\by,d\bu) &\triangleq \prod_{i=1}^N \lambda(y_i)(du_i) \otimes \mu_i(dy_i). \nonumber
\end{align}
Note that for any $(f_1,\ldots,f_N) \in \prod_{i=1}^N L_1(\mu_i,C_0(\sU_i))$, we have 
\begin{align}
&\lim_{n\rightarrow \infty} \int_{\sY \times \sU} \prod_{i=1}^N f_i(y_i)(u_i) \, \nu^n(d\by,d\bu) 
=\lim_{n\rightarrow \infty}\int_{\sY \times \sU} \prod_{i=1}^N f_i(y_i)(u_i) \, \lambda^n_i(y_i)(du_i) \,\, \mu_i(dy_i) \nonumber \\
&= \int_{\sY \times \sU} \prod_{i=1}^N f_i(y_i)(u_i) \, \lambda_i(y_i)(du_i) \,\, \mu_i(dy_i) \nonumber \\
&= \int_{\sY \times \sU} \prod_{i=1}^N f_i(y_i)(u_i) \, \nu(d\by,d\bu). \nonumber 
\end{align} 
Moreover, if 
\begin{align}
&\int_{\sY \times \sU} \prod_{i=1}^N f_i(y_i, u_i) \, \nu(d\by) = \int_{\sY \times \sU} \prod_{i=1}^N f_i(y_i, u_i) \, \hat{\nu}(d\by). \nonumber 
\end{align} 
for all $(f_1,\ldots,f_N) \in \prod_{i=1}^N C_c(\sY_i \times \sU_i)$, then 
$
\nu(d\by) = \hat{\nu}(d\by).
$
Since $\P_{\C}$ is tight, there exists a subsequence $\{\nu^{n_k}\}$ that converges to some probability measure $\hat{\nu}$. Note that since for any $f_i \in C_c(\sY_i \times \sU_i)$, $f_i(y_i,\cdot) \in L_1(\mu_i,C_0(\sU_i))$, we have 
\begin{align}
\lim_{k\rightarrow \infty} \int_{\sY \times \sU} \prod_{i=1}^N f_i(y_i,u_i) \, \lambda^{n_k}_i(y_i)(du_i) \,\, \mu_i(dy_i) &= \int_{\sY \times \sU} \prod_{i=1}^N f_i(y_i,u_i) \, \lambda_i(y_i)(du_i) \,\, \mu_i(dy_i) \nonumber \\ 
&= \int_{\sY \times \sU} \prod_{i=1}^N f_i(y_i,u_i) \, \nu(d\by,d\bu) \nonumber \\
&= \int_{\sY \times \sU} \prod_{i=1}^N f_i(y_i,u_i) \, \hat{\nu}(d\by,d\bu) \nonumber,
\end{align} 
where the last equality is due to weak convergence of $\nu^{n_k}$ to $\hat{\nu}$. Hence, $$\hat{\nu}=\nu=\prod_{i=1}^N \lambda_i(y_i)(du_i) \,\, \mu_i(dy_i)$$ is a probability measure and so 
for any $i=1,\ldots,N$, $\lambda_i(y_i) \in \P(\sU_i)$ $\mu_i$-a.e., which completes the proof. 
\end{proof}

\begin{definition}{(\cite[Definition 4.4]{GuYuBaLa15})}\label{aux4}
Let $\sE_1$, $\sE_2$, and $\sE_3$ be Borel spaces. A non-negative measurable function $\varphi: \sE_1 \times \sE_2 \times \sE_3 \rightarrow [0,\infty)$ is in class $\mathrm{IC}(\sE_1,\sE_2)$ if for every $M > 0$ and for every compact set $K \subset \sE_1$, there exists a compact set $L \subset \sE_2$ such that
\begin{align}
\inf_{K \times L^c \times \sE_3} \varphi(e_1,e_2,e_3) \geq M. \nonumber
\end{align}
\end{definition}

Using this definition, we now state the following result.

\begin{theorem}{(\cite[Lemma 4.5]{GuYuBaLa15})}\label{aux5}
Suppose $\varphi: \sE_1 \times \sE_2 \times \sE_3 \rightarrow [0,\infty)$ is in class $\mathrm{IC}(\sE_1,\sE_2)$. Let $m>0$ and $F_1 \subset \P(\sE_1)$ be a tight set of measures. Define
\begin{align}
&F =  \phantom{x}\biggl\{ \nu \in \P(\sE_1 \times \sE_2 \times \sE_3): \mathrm{Proj}_{\sE_1}(\nu) \in F_1 \text{ and } \int \varphi d\nu \leq m \biggr\}. \nonumber
\end{align}
Then $\mathrm{Proj}_{\sE_1 \times \sE_2}(F)$ is a tight set of measures.
\end{theorem}

\begin{theorem}
Suppose \emph{either }of the following conditions hold:
\begin{itemize}
\item [(i)] $\sU_i$ is compact for all $i$.
\item [(ii)] For non-compact case, we assume
\begin{itemize}
\item [(a)] The function $c(\,\cdot\,,x_0,\,\cdot\,,\,\cdot\,)$ is in class $\mathrm{IC}(\sY,\sU)$ for any $x_0 \in \sX_0$.
\item [(b)] For any compact $K \subset \sY$, $\inf_{(\by,x) \in K\times\sX} \prod_{i=1}^N q_i(y_i,x) > 0$.
\end{itemize}
\end{itemize}
Then, for all $r \in \R$, the set $\C_r \triangleq \{\lambda \in \Lambda: M(\lambda,x_0) \leq r\}$ is compact for any $x_0 \in \sX_0$; that is, $M(x_0,\,\cdot\,)$ is inf-compact for all $x_0 \in \sX_0$. Therefore, there exists an optimal team decision rule for $\mathbf{(SC)}$, and so, for $\mathbf{(S)}$. 
\end{theorem}

\begin{proof}
(i): If $\sU_i$ is compact for all $i$, then $\P_{\C_r}$ is automatically tight since marginal on $
\sY$ is fixed for any probability measure in $\P_{\C_r}$. Hence, by Lemma~\ref{prf1}, the closure of $\C_r$ is in $\Lambda$ and compact with respect to  $w^*$-topology as $\Lambda$ is relatively compact. But since $M(x_0,\,\cdot\,)$ is lower semi-continuous, $\C_r$ is closed. Hence, $\C_r$ is compact. \\
\noindent(ii):
We define $\tilde{c}(x,x_0,\by,\bu) \coloneqq c(x,x_0,\by,\bu) \prod_{i=1}^N q_i(y_i,x)$. Since for any compact set $K \subset \sY$, we have $\inf_{(\by,x) \in K\times\sX} \prod_{i=1}^N q_i(y_i,x) > 0$, $\tilde{c}$ is in class $\mathrm{IC}(\sY,\sU)$ for any $x_0 \in \sX_0$. Define
$$
F = \left\{\nu \in \P(\sY\times\sU\times\sX): \mathrm{Proj}_{\sY}(\nu) = \prod_{i=1}^N \mu_i \text{ and } \int \tilde{c} \, d\nu \leq r \right\}. 
$$
Then, by Theorem~\ref{aux5}, $\mathrm{Proj}_{\sY\times\sU}(F)$ is tight. But obviously, $\P_{\C_r}$ is a subset of $\mathrm{Proj}_{\sY\times\sU}(F)$ as $M(x_0,\,\cdot\,) \leq r$ on $\C_r$, and so, $\P_{\C_r}$ is tight. Hence, by Lemma~\ref{prf1}, the closure of $\C_r$ is in $\Lambda$ and compact with respect to  $w^*$-topology as $\Lambda$ is relatively compact. But since $M(x_0,\,\cdot\,)$ is lower semi-continuous, $\C_r$ is closed. Hence, $\C_r$ is compact.
\end{proof}

\begin{remark}
As discussed in the introduction, a similar existence result can be achieved by applying the strategic measure approach introduced in \cite{YuSa17} to team problems. In this approach, when we have a common observation $x_0 \in \sX_0$, the problem can be viewed as a classical totally decentralized static team problem. By utilizing the topology defined in \cite{YuSa17} for strategic measures, we can demonstrate the existence of an optimal strategic measure $P_{x_0}$ for any $x_0$, subject to similar assumptions as those presented in this paper. Then, employing a measurable selection theorem, we can establish that the mapping $x_0 \rightarrow P_{x_0}$ is measurable. Consequently, $\mu_0(dx_0) \otimes P_{x_0}(dx,d\by,d\bu)$ emerges as the optimal strategic measure for the observation-sharing team problem.
\end{remark}

\section{Conclusion}

In conclusion, this paper has explored the existence of team-optimal strategies for static teams operating within observation-sharing information structures. Our approach, assuming that agents having access to shared observations, involves the transformation of the team problem into an equivalent centralized stochastic control problem through the introduction of a topology on policies. By employing conventional stochastic control techniques, we have successfully demonstrated the existence of team-optimal strategies. This research extends the well-established common information approach, initially tailored for discrete scenarios, to a more abstract continuous framework. The central challenge in this endeavor is the identification of the most suitable topology on policies.

\section{Acknowledgment}
The author is grateful to Serdar Y\"{u}ksel for his constructive comments and the late Ari Arapostathis for a very informative discussion on topologies on Markov policies.


\end{document}